\documentclass[12pt,a4paper]{article}
\usepackage[utf8]{inputenc}
\usepackage[T1]{fontenc}
\usepackage[round]{natbib}
\usepackage{enumitem}
\usepackage[centertags]{amsmath}
\usepackage{amsfonts,amsthm,amssymb}
\usepackage{amssymb}
\usepackage{amsmath}
\usepackage{url}
\usepackage{soul}
\usepackage{graphicx}
\usepackage{accents}
\usepackage{booktabs}
\usepackage{appendix}
\usepackage[hmargin=3.175cm,vmargin=2.175cm]{geometry}
\usepackage{etoolbox}
\usepackage{xparse}
\usepackage{enumitem}
\usepackage{authblk}
\usepackage{color}
\usepackage{soul}
\usepackage{mathtools}
\usepackage{chngcntr}
\usepackage{authblk}
\usepackage{placeins}
\usepackage[percent]{overpic}
\usepackage{enumitem}
\usepackage{pgfplots}
\usepackage{pgf}
\usepackage{caption}
\usepackage{subcaption}

\newcommand{\ubar}[1]{\underaccent{\bar}{#1}}
\DeclareDocumentCommand{\publicBelief}{O{\mu}}{#1}
\DeclareDocumentCommand{\privateBelief}{O{p}}{#1}
\DeclareDocumentCommand{\signalupdate}{O{q}}{#1}
\DeclareDocumentCommand{\signal}{O{s}}{#1}
\DeclareDocumentCommand{\signalsSet}{O{S}}{#1}
\DeclareDocumentCommand{\state}{O{\omega}}{#1}
\DeclareDocumentCommand{\statesSet}{O{\Omega}}{#1}
\DeclareDocumentCommand{\price}{O{\tau}}{#1}
\DeclareDocumentCommand{\action}{O{a}}{#1}
\DeclareDocumentCommand{\limitParam}{O{\alpha}}{#1}
\DeclareDocumentCommand{\deterrencePrice}{O{\price} O{d}}{#1^#2}
\DeclareDocumentCommand{\LLR}{O{x}}{\log(\frac{#1}{1-#1})}
\DeclareDocumentCommand{\llr}{O{x}}{l\left(#1\right)}
\DeclareDocumentCommand{\lBound}{O{\limitParam} O{\publicBelief}}{\ubar{#1}_{#2}}
\DeclareDocumentCommand{\uBound}{O{\limitParam} O{\publicBelief}}{\bar{#1}_{#2}}
\DeclareDocumentCommand{\eqPrice}{O{\price}}{#1^{*}}

\newtheorem{lemma}{Lemma}
\newtheorem{proposition}{Proposition}

\newtheorem{theorem}{Theorem}

\theoremstyle{definition}
\newtheorem{definition}{Definition}

\newlist{secenum}{enumerate}{10}
\setlist[secenum]{label=\thesection.\arabic*,leftmargin=*}

\newtheorem{innercustomgeneric}{\customgenericname}
\providecommand{\customgenericname}{}
\newcommand{\newcustomtheorem}[2]{%
	\newenvironment{#1}[1]
	{%
		\renewcommand\customgenericname{#2}%
		\renewcommand\theinnercustomgeneric{##1}%
		\innercustomgeneric
	}
	{\endinnercustomgeneric}
}

\newcustomtheorem{customthm}{Theorem}
\newcustomtheorem{customlemma}{Lemma}
\newcommand{\blocktheorem}[1]{%
	\csletcs{old#1}{#1}
	\csletcs{endold#1}{end#1}
	\RenewDocumentEnvironment{#1}{o}
	{\par\addvspace{1.5ex}
		\noindent\begin{minipage}{\textwidth}
			\IfNoValueTF{##1}
			{\csuse{old#1}}
			{\csuse{old#1}[##1]}}
		{\csuse{endold#1}
		\end{minipage}
		\par\addvspace{1.5ex}}
}

\blocktheorem{customlemma}

\definecolor{ao}{rgb}{0.0, 0.5, 0.0}

\usepackage{hyperref}
\usepackage{subcaption}
\usepackage{authblk}

\usepackage[T1]{fontenc}
\usepackage{titling}

\pretitle{\begin{center}\Huge\bfseries}
\posttitle{\par\end{center}\vskip 0.5em}
\preauthor{%
  \begin{center}
    \large \lineskip 0.5em%
    \scshape
    \begin{tabular}[t]{c}%
}
\postauthor{\end{tabular}\par\vspace{-3mm}\end{center}}
\predate{\begin{center}\normalsize}
\postdate{\par\end{center}}

\begin{document}

\title{The Gatekeeper Effect:  The Implications of Pre-Screening, Self-selection, and Bias for Hiring Processes.}

\date{\today}
\author{Moran Koren}
\maketitle

\begin{abstract}

We study the problem of screening in decision-making processes under uncertainty, focusing on the impact of adding an additional screening stage, commonly known as a 'gatekeeper.' While our primary analysis is rooted in the context of job market hiring, the principles and findings are broadly applicable to areas such as educational admissions, healthcare patient selection, and financial loan approvals.

The gatekeeper's role is to assess applicants' suitability before significant investments are made. Our study reveals that while gatekeepers are designed to streamline the selection process by filtering out less likely candidates, they can sometimes inadvertently affect the candidates' own decision-making process. We explore the conditions under which the introduction of a gatekeeper can enhance or impede the efficiency of these processes. Additionally, we consider how adjusting gatekeeping strategies might impact the accuracy of selection decisions.

Our research also extends to scenarios where gatekeeping is influenced by historical biases, particularly in competitive settings like  hiring. We discover that candidates confronted with a statistically biased gatekeeping process are more likely to withdraw from applying, thereby perpetuating the previously mentioned historical biases. The study suggests that measures such as affirmative action can be effective in addressing these biases.

While centered on hiring, the insights and methodologies from our study have significant implications for a wide range of fields where screening and gatekeeping are integral.



%
\end{abstract}
\newpage
Hiring a new employee is a costly and noisy process not only for the firm, which must examine the quality of job applicants, but also for the applicants, who incur travel costs and invest of their own time to prepare for and take the prospective employee assessment test. The costs of this process can also be significant for the firm due to the amount of resources that it must allocate to search for and interview potential candidates. Additionally, the level of noise that is typically generated by the process can be disruptive to both firm and applicant, because typically, the firm must sift through many applications, and the applicants must endure a long and intense hiring process to prove their worth. In the United States, for example, the average company spends about \$4,000  to hire a new employee in a process that lasts up to 52 days to fill a position \citep{vasconcellosHowMuchDoes2023}. From the applicant's perspective, applying for a new job is also physically demanding, e.g., the average job search, from the submission of applications to hiring, takes about five to six months and involves several rounds of interviews and/or tests \citep{skallerupbessetteCostJobApplication2013}.

From the firm's perspective, to reduce the hiring costs that they incur, selection processes increasingly contain a pre-filtering stage called \textit{a gatekeeper.} The gatekeeper predicts each applicant's probability of success based on her observable traits and attempts to separate the wheat from the chaff by testing only those applicants who are most likely to succeed. While the role of gatekeeper in hiring processes has often been filled by an HR representative, that task has been automated in recent years by using AI \citep{windleyCouncilPostAI,greenfieldArtificialIntelligenceComing2018}. The automation of hiring processes, however, is a challenge for two main reasons \citep{councilCouncilPost10,greenfieldArtificialIntelligenceComing2018,schildhornHowAIDeciding2022}. First, insofar as an applicant’s fit for a position depends on immeasurable characteristics, such as demeanor, diligence, and intelligence, the task of finding a suitable candidate for a specific position is difficult to generalize. Second, the statistical models used for such automation often assume that observations are independently drawn from an invariant population. As we claim, however, in the context of hiring, this assumption is violated due to candidate self-selection. While the first challenge may be mitigated by estimating the candidate's fit based on observable traits (e.g., candidate's education, previous experience, etc.), the latter effect has significant implications for the quality and efficiency of the selection process as a whole.

Intuitively, the introduction of a gatekeeper will reduce the firm's selection costs as it prevents candidates who are less likely to be hired from taking the test. However, as we show, the gatekeeper incurs an indirect consequence. By pre-rejecting candidates who are less likely to succeed in the test, the gatekeeper interferes with candidate self-selection. To see this indirect effect, consider a candidate who contemplates applying.   
Before applying for the position and incurring the relevant costs (test preparation, travel arrangements, etc.), each candidate evaluates her odds of eventually being selected to fill the open position. Naturally, she will choose to opt-out when, based on her own assessment, the odds of her being selected are too low. When the selection process features a gatekeeper, the candidate knows that she will incur the costs only if the gatekeeper's information is in her favor.\footnote{One example of such a cost structure can be found in the academic job market, where candidates who pass the initial screening are invited to fly out to the company's location for in-person interviews. Similar incentives also emerge when prospective students are considering to which graduate programs they should apply. The potential future students do not necessarily know how grad school will turn out for them (but they do know their own capabilities), and the universities do not have an accurate assessment of the abilities of each candidate. I thank the anonymous reviewers for these examples.} In the event that the information possessed by the gatekeeper about the job candidate is positive, the marginal candidate would prefer to opt in and take the prospective employee assessment test, effectively lowering the average quality of applications. We asked whether this indirect effect has the potential to override the intended gatekeeper effect. In other words, can the introduction of a gatekeeper reduce the correctness of a selection process? And if it can, under what conditions? 

To answer these questions, we present a game of incomplete information that is played between a candidate and a gatekeeper and that comprises two states of nature: in one, the candidate and the position are highly compatible, while in the other, the fit between candidate and position is poor. Both agents receive a noisy signal about the identity of the real state, yet the state of nature is unknown. In this three-stage game, the candidate decides first whether to apply. Second, the gatekeeper then decides whether to submit the candidate to a costly, prospective employee assessment test. If the candidate passes the gatekeeping process, she incurs subsequent testing costs and undergoes a noisy employee assessment test. Upon passing the employee assessment test, she is hired. We start our game by assuming that the gatekeeper is technocratic and that it follows its signal mechanically. 
In this scenario, (1) the introduction of a gatekeeper interferes with the candidate self-selection process and lowers the average quality of applications; (2) we characterize the conditions under which the introduction of a gatekeeper reduces the overall correctness of the selection process;  (3) we then allow the gatekeeper to behave strategically, i.e., we assume that it disregards its signal with positive probability. In this newly introduced game between an applicant and a strategic gatekeeper, we show that, in some cases, the gatekeeper can rectify the problem caused by its potential interference with candidate self-selection by playing a mixed strategy, thus restoring the system's efficacy.

In the context of statistical learning, the gatekeeper effect is another example of the base rate fallacy because it skews self-selection. As a result, the overall quality of the applicant pool may decrease, as some of the applying candidates may be overestimating their fit after they ignored the base rate of what typically constitutes a successful candidate. If gatekeepers are aware of this fallacy, they can adjust their behavior strategically. By considering the base rates (which are represented by the prior in our model) more effectively in their decision-making process, they could mitigate the negative impact on candidate self-selection and improve the overall accuracy of the selection process.

A highly relevant application of our model is to study the effects of existing biases in learning environments. One typical environment in this context is the hiring process, in which efforts are increasingly aimed at automating the early stages of that process. Perhaps the main concern that is voiced in light of the trend of using AI to automate gatekeepers is whether a reliance on such automation will fortify and exacerbate existing biases due to the scarcity of training data.\footnote{In a recent example, Amazon implemented an AI recruitment system that was summarily scrapped when researchers realized that it was biased against women (see \url{https://www.bbc.com/news/technology-45809919}). Analysis showed that the bias occurred due to the scarcity of women's CVs among the cohort in the system.} We study the gatekeeper effect of algorithmic bias by extending our model to a setting in which two candidates compete for a single position. Due to historical reasons, the gatekeeper's signal for one of the candidates is more accurate than its signal for the other. In this setting, we show that the candidate whose signal is less accurate will be more hesitant than his competitor. In a dynamic setting, this behavior will effectively increase the aforementioned gap in the accuracy of the prospective employee test. We use an example of a specific information structure to demonstrate that the latter effect can be mitigated by affirmative action; given equal ex-post candidate types, there is a higher probability that the candidate for whom the information is less accurate will be hired.

\noindent\textbf{Organization:} In Section \ref{sec:rel}, we describe the related work. Our model, analysis, and main results are presented in Section \ref{sec:model}. We tackle strategic gatekeeping in Section \ref{sec:strategic}. We study the biased gatekeeper effect in Section \ref{sec:biasedGK}. We demonstrate the implications of our results by using an example in Section \ref{sec:example}. Section \ref{sec:discu} contains a discussion about the applicability of our results to other settings, such as peer review. Finally, in Section \ref{sec:conc}, we present our conclusions.

\section{Related work}\label{sec:rel}
The impacts of the strategic behaviors of rational agents vis-a-vis economic efficiency in the presence of uncertainty has been studied extensively during the last half century.  In his seminal paper,  \citet{Akerlof1970} presented a market model wherein the value of a proposed good (a used car) is known only by one side of the market, while the other side is oblivious to the good's value. He showed that in such markets, the average quality of the proposed commodity decreases due to adverse selection (e.g., ``lemons" are driving good cars out of the market).  In follow-up work,  \citet{Levin2001a} showed that adverse selection depends on the relative quality of the information that is available to both sides of the market, namely, while holding the information structure of the buyer fixed, a more informed seller increases the severity of the so-called buyer's curse and vice versa.   Our model diverges from this line of thinking in three ways. First, it contains no asymmetric information. In our model, both sides are oblivious to the underlying state, and the quality of player signals may vary. Second, much of the literature on adverse selection assumes an inherent conflict of interest. In our model, if agents hold perfect information, the incentives of both players are fully aligned. The tension between the players is due to the uncertainty that they both face. Finally, our results are not driven by an exogenously determined price, although the applicant cost parameter plays a similar role. 

In addition, our results are related to a growing line of research on screening. \citet{Lagziel2019} studied the problem of screening efficiency, namely, whether a stricter filtering step improves the quality of the resulting decision. The authors found a non-monotonic connection between the two. There are cases,for example, in which a higher screening threshold induces a reduction in the overall expected valuation.
Moreover, \citet{Lagziel2019}  characterized the conditions of the distribution under which such non-monotonic behavior emerges.   In a follow-up paper,  \cite{Lagziel2021} showed that a noisily informed decision-maker can improve the quality of his decision by introducing additional binary noise. Our results are similar, and we show that, in some cases, one can improve correctness even by adding a gatekeeper whose signal is extremely noisy and  economically equivalent to a coin toss. In the other cases, the gatekeeper's signal quality determines the added value of an additional filtering step. Our approach, however, diverges from that in the existing screening literature, as the driving force of our model is the agents' strategic choice, while in their model, in contrast, the primary focus is on characterizing the conditions of the probability structure under which additional ``mechanical" screening is beneficial.

In recent years, the importance of self-selection to the efficiency of the economic system when agents are strategic has gained traction  \citep{hallAmplificationUnemploymentFluctuations2005,leeTheorySelfselectionMarket2009,panInstabilityMatchingOverconfident2019}.  \cite{hallAmplificationUnemploymentFluctuations2005} studied the influence of friction on self-selection in labor markets, focusing on the role played by the firm's candidate evaluation costs. \cite{leeTheorySelfselectionMarket2009} and \cite{panInstabilityMatchingOverconfident2019} studied self-selection in noisy matching markets. We contribute to this growing literature by uncovering the inverse relationship between signal accuracy and self-selection. To the best of our knowledge, we are the first to present such a model. Furthermore, we show that the phenomenon we uncovered is evident in competitive settings in which one agent is presented with a less accurate signal, a case that is crucial to the design and implementation of AI based allocation systems.  In this work, we supplement \cite{hallAmplificationUnemploymentFluctuations2005} by  focusing on the informational dimension of the selection process as it manifests in a pre-filtering stage rather than examining its macroeconomic implications. We do so by using an information structure borrowed from the social learning literature. To the best of our knowledge, this work marks the first attempt to answer questions about information aggregation in such markets. Our findings are similar to those found in the experimental examination of  \cite{panInstabilityMatchingOverconfident2019}, where candidate over-participation leads to instability. In our case, however, the over-participation occurs during an earlier stage, before the candidates learn the result of the pre-filtering stage, and it is motivated only by the selection process design.  \cite{leeTheorySelfselectionMarket2009} proved that lengthy review times may decrease the rate at which new research is published, also due to self-selection, but he concluded that ultimately, the manuscript peer review process also improves the overall quality of publication. In the context of our model, this is akin to increasing the candidate's testing cost. We show that while the typically long durations of the review times indeed decrease participation, we prove that it can also cause a reduction in the overall selection correctness. Finally, while previous work focused on market wide performance measures such as stability and welfare, we chose to focus on the accuracy of the resulting decision, which we call correctness. This choice was motivated by two factors: it is a natural measure with which to analyze the performance of the algorithmic decision support tool, and it is well defined, even within the narrow scope of an individual decision (i.e., not considering the entire market). 

The study of algorithms in the presence of bias has gained traction recently (see \citealp{smithAlgorithmsBias2021} for a recent review). \cite{danksAlgorithmicBiasAutonomous2017} analyzed the implications of bias in various algorithmic learning settings and autonomous agent design scenarios, concluding that the problem of algorithmic bias cannot be solved for the general case, and therefore, treatments for the bias effect should be tailored to the specific scenario.

Recent research in the field of artificial intelligence (AI) has illuminated the crucial role of data generation processes in AI training tasks. \citet{Tiwald2021SyntheticData} highlight the need for synthetic data that is both representative and fair as a vital component in the responsible training of AI systems. This is especially pertinent in areas like hiring, wherein AI-driven decision making has profound societal implications. \citet{Mazilu2020DataWrangling} also stress the importance of fairness in data wrangling, advocating for diverse and unbiased data sets in AI training, which is key to preventing the perpetuation of existing biases in AI applications.

Examining the ethics of AI, \citet{Ntoutsi2020EthicalAI} and \citet{Lepri2021HumanCentricAI} delve into the moral imperatives of AI system design. They argue for the integration of ethical and legal principles from the inception of AI systems, thereby ensuring that these technologies will benefit society as a whole. Such a responsible approach is essential to maintain public trust in AI. Echoing these sentiments, \citet{Cahan2019PredictiveModels} focus on the utility, equity, and generalizability of predictive AI models, underscoring the importance of population-representative training data.

 In the context of our examination of statistical bias in hiring, \cite{cowgill2020bias} and \cite{liHiringExploration2020} study the effectiveness of machine learning in hiring, demonstrating that AI technologies can mitigate the biases that lurk in traditional human decision-making processes. The \cite{cowgill2020bias} study presents a model in which machine learning and human experimentation complement each other, showing that algorithms can reduce biases, particularly when human decisions are inconsistent. This was evidenced in a field experiment that tested for biases when hiring for white-collar jobs. In that experiment, the use of machine learning led to improved candidate selection that benefited non-traditional candidates like women, minorities, and those from non-elite colleges. Similarly, \cite{liHiringExploration2020} frames the recruitment process as a contextual bandit problem, in which balancing exploitation with exploration is key. Their exploration-based algorithm not only increased the quality of candidates selected for interviews, it also enhanced demographic diversity and outperformed traditional algorithms. 
\cite{kasinidouAgreeDecisionThey2021} shifts the focus to the perception of algorithmic fairness among the potential future developers of such systems. Through an online survey of students in fields adjacent to algorithm development, the study explores how these future developers perceive fairness in different algorithmic decision-making scenarios. Among its key findings, the study showed that agreement with an algorithmic decision does not necessarily correlate with the fairness of the outcome, and that participants view the proportional distribution of benefits as fairer. The study highlights the complexity of fairness in algorithmic decision-making and underscores the importance of understanding both the technical and ethical dimensions of algorithm development, especially from the perspective of future algorithm developers. 

Our study complements this research direction in the literature in two ways. First, we provide a theoretic model describing the effects of bias in hiring processes (or other noisy selection processes) and show that such bias influences self-selection and degrade process efficiency. Second, we suggest that skill-based affirmative action be leveraged to eliminate the inefficiencies that occur.

Our study of strategic gatekeeping is also reminiscent of the work on Bayesian Persuasion of \cite{Kamenica2011}. In that work, a fully informed sender gains additional rewards by employing a strategy under which he lies to the receiver at some positive probability (See \citealp{Kamenica2019} for a recent review). In our case, the gatekeeper's dishonesty incentivizes more selective candidate behavior. We show that the benefits of this strategy are guaranteed when the public information is sufficiently favorable.

Finally, we borrow our information structure from the literature on information cascades and herding \citep{Banerjee1992,bikhchandani1992theory,Smith2012,Arieli2019}. In this group of models, sequentially arriving agents learn from (or herd on) the actions of agents who arrived before them. In most herding models, the strategic interaction between agents is limited because agents are assumed to exist for a single period and are unaffected by future events (see \citealp{NBERw28887} for a recent survey).\footnote{Recently, papers that either highlight the strategic game between the agents or consider the overall effect have emerged.  See, for example, \citealp{ban2020sequential,Arieli2018a,Halac2020,Smith2021}.} In our model, the strategic interaction between both agents is the focal point of our analysis.

\section{The Model}\label{sec:model}

The game comprises a candidate ($C$) and a gatekeeper ($K$).  This is a three stage game, in which first nature randomly selects a hidden value, and both the candidate and the gatekeeper receive a noisy signal about the realized valuation, which describes the candidate's fit for the proposed position. In stage 2, the candidate sees her signal and decides whether to apply for the position. If she opts out, her utility is zero, and if she decides to apply, she becomes an applicant and faces the gatekeeper. In stage 3, the gatekeeper sees its signal over the applicant's fit and decides whether to allow her to take the test. If the gatekeeper positively assesses the applicant, she incurs a cost and takes a noisy test. If she passes the test, she is hired and her type is revealed. Next we provide formal descriptions of the game and of the utilities of the agents.

There are two possible states of nature that we denote by $\Omega=\{H,L\},$ wherein state $H$ signifies that the candidate is highly compatible with the position, and state $L$ signifies that there is a low level of compatibility between the candidate and the position.

To obtain the offered job position, the candidate must pass a costly test. In state $\omega\in\Omega=\{H, L\},$  the probability that a candidate will pass the test is denoted by $\varphi(\omega).$ We assume that $\varphi(H)>\varphi(L)>0.$ We normalize $\varphi(H)=1$ and denote $\varphi(L)\equiv\varphi.$\footnote{Note that this assumption is without loss of generality. Our only limitation is that the above inequalities be strict.}  If she decides to apply and is allowed to take the test,  the candidate incurs a cost $\gamma\in(\varphi,1)$ that must be paid in the event that the test takes place, even if the candidate does not pass it.

The candidate can decide whether to apply, $a_C\in\{0,1\}.$  If she decides not to apply, she receives a utility of zero, but if she applies ($a_C=1$) and passes the test, she obtains a positive utility (even when $\omega=L).$  We normalize the utility of a high fitting candidate to $1$ and assume that the utility of a low-fitting candidate is $\alpha>0,$ such that $\alpha\varphi<\gamma.$  The gatekeeper can decide whether to submit the applicant to the test, $a_K=\{0,1\}.$ It receives a utility of $1$ whenever a highly fit candidate passes the test and endures a utility of $-d$ whenever a candidate with a low fit passes the test, where $d>0.$\footnote{Note that the assumption about the gatekeeper's utility is without loss, as any cost structure can be fitted by choice of an appropriate $d$.}  

The realized state is unknown. Both players assign a prior probability of $\mu$ to the event $Pr(\omega=H).$ We will call $\mu$ the \textit{public belief}.     In addition, each player $i\in\{C,K\}$ receives a noisy signal $s_i\in S_i.$ We assume that the gatekeeper's signal is binary, i.e., $S_K=\{h,l\}.$\footnote{Most of our results carry through even when we assume a richer information set. We discuss this further in Section \ref{sec:conc}.}   The candidate draws her signal from  an open interval over $S_C\subset\mathbb{R}.$    
 The quality of the gatekeeper's signal $Pr(s_k=\omega|\omega)=q_K>\frac{1}{2}$ is known (hereafter, we suppress the subscript $K$).  The candidate's signal is randomly drawn from a state-dependent distribution $F_\omega$ with PDFs $f_\omega.$ The distributions $F_L, F_H$ 	are mutually absolutely continuous, i.e., no signal fully reveals the unknown state.  Additionally, we make two simplifying assumptions. First we assume that $F_L, F_H$ exhibit the Monotone Likelihood Ratio Property (MLRP), which is defined as follows,
 
 \begin{definition} $F_L,F_H$ exhibit the \emph{Monotone Likelihood Ratio Property}  if for every $s,\hat s\in S_C$ such that $s>\hat s,$ the following condition is satisfied,  $$\frac{f_H( s)}{f_L( s)}\ge\frac{f_H(\hat s)}{f_L(\hat s)}.$$
\end{definition}
\noindent Second, we assume that the candidate signals are unbounded. The following definition of unbounded signals is borrowed from \cite{Smith2012}, 
\begin{definition}
Let $F_H,F_L$ be two mutually absolutely continuous distributions with a common support $supp(F).$ We say that the signals are \textit{unbounded} if the convex hull $co(supp)=[0,1].$
\end{definition}
The implications of assuming an unbounded signal structure are that for every $\mu\in(0,1)$ and every $q_K\in[0.5,1),$ the candidate will always take either action with a positive probability. In social learning models, this creates the \textit{overturning principle}, which facilitates learning in infinite populations. In our context, it is a means to maintain clarity. Once it has been combined with MLRP, it guarantees that the candidate follows a unique threshold strategy, in which both actions occur with positive probability. We believe that releasing this assumption will add little insight and will not alter the qualitative nature of our result.\footnote{We discuss this modeling choice further in Section \ref{sec:conc}.}

 Let $\sigma$ denote a strategy profile for both players. Note that, together with the game's information structure, $\sigma$ defines a probability distribution $\mathbb{P}$ over $\Omega\times S_C \times S_K.$ Therefore, given a strategy profile $\sigma,$  the players' expected utilities are well defined and can be calculated.

\subsection{Analysis and Results}\label{sec:analysis}

Our first question is, ``What are the effects of embedding a `Mechanical gatekeeper' in a selection process on process efficiency?''
 In other words, can we always improve the quality of the selection process by augmenting it with a gatekeeper that automatically follows its signal and that is deprived of any strategic considerations?  
Intuitively, one would expect the answer to depend on the gatekeeper's signal quality. If the gatekeeper's signal is of sufficiently high quality, the resulting allocation will improve, but it may be reduced if a low-quality gatekeeper is used.  Notwithstanding the merit of this intuition, we will show that it is not always correct.

A \textit{Mechanical} strategy for the gatekeeper is defined by $\sigma_K(h)=1,\sigma_K(l)=0.$ A \textit{Mechanical Gatekeeper} is a gatekeeper whose strategy space is restricted to playing only the mechanical strategy.  We will use the subscript $MK$ to denote a ``mechanical keeper.''

 The candidate uses Bayes rule to update her posterior belief. Upon receiving a signal $s\in S_C,$ her expected utility from applying will be,

\begin{equation}\label{eq:candidate_utility}
	\begin{split}
	&\text{\footnotesize $U_C(a=1|s)=$}\\ 
	& \text{\footnotesize $\frac{\mu p(s)}{\mu p(s)+(1-\mu)(1-p(s))}\alpha^H_{\sigma_K}(1-\gamma)+\frac{(1-\mu)(1-p(s))}{\mu p(s)+(1-\mu)(\alpha-p(s))}\alpha^L_{\sigma_K}(\alpha\varphi-\gamma),$}
\end{split}
\end{equation}

where $\alpha^\omega_{\sigma_K}:=Pr_\sigma(a_K=1|\omega)$ is the probability of passing the gatekeeper when she plays strategy $\sigma_K$ and the realized state is $\omega,$ and $p(s)=\frac{f_H(s)}{f_H(s)+f_L(s)}.$ 
Note that we can map each signal to its respective $p(s).$ Thus, without loss of generality, we assume that the signal $s$ admits the posterior $p(s),$ and all signals that share the posterior are grouped.  Note that $p(s)$  measures the expected level of the candidate's fit based on her private information when all other factors are held constant. Therefore, we denote \textit{the candidate's 
 subjective quality} with $p(s)$ and make the following definition.
 


\begin{definition}
The  candidate's \textit{subjective quality} upon receiving signal $s\in S_C$ will be denoted by $p(s)\triangleq \frac{f_H(s)}{f_H(s)+f_L(s)}.$  
\end{definition}
Our assumptions about the signal structure translate to $\{x: s \in S_C \mbox{ and } x=p(s)\}=[0,1].$ 
Let $x(\sigma_K)=\inf\{x: x=p(s)\mbox{ and }s\in S_C\mbox{ and }u_C(a_C=1|s,\sigma_K)>0\}.$ That is, when the gatekeeper's strategy is $\sigma_K,$ $x(\sigma_K)$ is the lowest signal for which the candidate's optimal action is to apply. As we assume strict MLRP, we know that for every $\sigma_K,$ $x(\sigma_K)$ is uniquely defined. Furthermore, one can map every signal $s$ to the posterior it induces $p(s).$ Therefore, we use these terms interchangeably in our analysis.  We will use the term \textit{applicant} to describe a candidate who decides to apply.  Under our assumptions, applicants are candidates whose subjective quality is greater than $x(\sigma_K).$ 
We  now present our first result, which refers to the mechanical gatekeeper.

\begin{theorem}\label{thm:selfSelect}
When facing a mechanical gatekeeper, the average subjective quality of the applicant decreases with the gatekeeper's signal quality.
\end{theorem}
\begin{proof}
	 Let $F_L,F_H,S$ be an information structure and let $q,\hat q$ be two signal qualities, such that $q>\hat q.$   Given a mechanical gatekeeper with signal quality $q$ and by equation \eqref{eq:candidate_utility}, a candidate will apply whenever the following condition is satisfied,
	 \begin{equation}\label{eq:applicant_condition}
		\frac{p(s)}{1-p(s)}\ge \frac{\gamma-\alpha\varphi}{1-\gamma}\frac{1-\mu}{\mu}\frac{1-q}{q}.	 	
	\end{equation}
	The threshold $x(\sigma_{MK}(q))$ is the subjective quality for which equation \eqref{eq:applicant_condition} holds with equality, that is, the signal that yields the lowest posterior belief of a high fit and for which the candidate still decides to apply.  Note that under our assumptions of MLRP and the possible set of subjective qualities, this threshold is uniquely defined for every $q.$  To complete the proof of the theorem, note that the function $\frac{y}{1-y}$ is decreasing.
\end{proof}

In Theorem \ref{thm:selfSelect}, we prove that the gatekeeper has an indirect effect on the candidate's choice. In addition to filtering out the less likely applicants, it also interferes with candidate self-selection. Candidates whose subjective quality was marginal and who would otherwise be deterred by the testing costs, will now apply, thinking that the gatekeeper's decision will protect them against unnecessary costs (as the event in which these costs materialize but the candidate fails becomes less likely). However, the question remains: can this indirect effect dominate?

In this work, we examine the performance of a selection process. Our measure for such performances is the accuracy of the resulting decision, that is, the probability of reaching the correct decision.  For the notion of correctness, we use that suggested by \cite{arieli2018one}. The correctness of the game, denoted by $\theta$, is defined as,
\begin{equation}
	\theta=Pr(H)Pr(\psi|H)+Pr(L)(1-Pr(\psi|L)),
\end{equation}
where $\psi$ is the event in which the candidate passes the test (i.e., the candidate applies, is approved by the gatekeeper, and passes the test).  In what follows, one can think of a system designer who wishes to maximize the selection process's correctness.\footnote{While we assume that the designer's utility is affected by false positives and false negatives of the same magnitude, we believe that releasing this assumption will negligibly affect the qualitative nature of our results. We release this assumption when we discuss the strategic gatekeeper.}

In our second result, we discuss the influence of introducing a gatekeeper on the correctness of the selection process. In Theorem \ref{thm:correctness implications}, we compare the correctness of a selection process in which no gatekeeper is present to the case of a mechanical gatekeeper with signal quality $q.$ Let $\hat x$ denote the candidate's threshold type when no gatekeeper is present. 
\begin{theorem}\label{thm:correctness implications}
	For any information structure $F_H,F_L,S,\mu,$
	\begin{enumerate}
		\item  there exists $\bar{q}$ such that for all $q>\bar q,$ introducing a mechanical gatekeeper with signal quality $q$ improves correctness.
		\item 	If $\frac{\mu(1-F_H(\hat x))}{(1-\mu)(1-F_L(\hat x))}>\varphi,$ then there exists  $\ubar{q}$ such that whenever $q<\ubar{q}$ the introduction of a mechanical gatekeeper lowers the  correctness of the selection process. 
		\item Let $\bar{q}>0.5.$ If $\frac{\mu(1-F_H(\hat x))}{(1-\mu)(1-F_L(\hat x))}<\varphi,$ then the introduction of a mechanical gatekeeper of quality $q\in (0.5,\bar{q})$ improves the correctness of the selection process.
	\end{enumerate}
\end{theorem}
The first and second parts of Theorem \ref{thm:correctness implications} state that a high-quality gatekeeper is always beneficial, while a low-quality gatekeeper can be detrimental. The third part is somewhat surprising, stating that correctness can improve even with the introduction of a gatekeeper of arbitrarily low quality. For example, take an arbitrarily small $\varepsilon.$ A gatekeeper whose signal quality is  $\frac{1}{2}+\varepsilon$ is economically equivalent to a coin toss.  Despite this equivalence, part three of the theorem tells us that some selection processes benefit from the introduction of a gatekeeper into the system.  

The intuition behind this result is as follows. Consider an alternative model in which the candidate sees the gatekeeper's signal before making her decision, and therefore, she can use it to help her decide her course of action. If the candidate's signal is low, her decision to opt out will not be affected by the additional information, but if the candidate's private signal is intermediate, her best response is to follow the gatekeeper's signal. The correctness in this instance will be equal to that of the original model. When the candidate's signal is sufficiently strong, however, the candidate's optimal choice is to disregard the new information from the gatekeeper and apply for the position, even if the gatekeeper's signal is negative. To maximize correctness, the candidate should thus disregard the gatekeeper's signal for this group of signals, because they are more accurate. The mechanical gatekeeper, however, does not allow it, and therefore, the introduction of a gatekeeper decreases the correctness when the candidate's signal is sufficiently strong.  The higher the signal quality of the gatekeeper, the smaller the set of candidate signals that are negatively affected. An additional effect comes into play when the gatekeeper's signal is sufficiently noisy. In the noisy case, the reduction in overall correctness due to the participation of candidate types who receive strong signals is countered by the increased correctness due to the increased selectivity vis-a-vis intermediate type candidates.


\begin{proof}
First we calculate $\hat x.$ By equation \eqref{eq:candidate_utility}, we  find that  $\hat x=\frac{(\gamma-\alpha\varphi)(1-\mu)}{(\gamma-\alpha\varphi)(1-\mu)+(1-\gamma)\mu}.$
By the definition of correctness, when no gatekeeper is introduced, the correctness will be
\begin{equation}
	\begin{split}
		\hat\theta&=\mu (1-F_H(\hat x)) +(1-\mu)(F_L(\hat x)+(1-F_L(\hat x))(1-\varphi))=\\
		&\mu (1-F_H(\hat x)) +(1-\mu)(1-(1-F_L(\hat x))\varphi).
	\end{split}
\end{equation} 
 When introducing a mechanical gatekeeper with signal quality q, the process correctness is,
$$
\theta(q)=\mu q (1-F_H(x(q)))+(1-\mu)(1-(1-F_L(x(q)))(1-q)\varphi)
$$
where
\begin{equation}\label{eq:xq}
x(q)=\frac{(\gamma-\alpha\varphi)(1-\mu)(1-q)}{(\gamma-\alpha\varphi)(1-\mu)(1-q)+(1-\gamma)\mu q}.
\end{equation}
We calculate the difference between both expressions and find:
\begin{equation}
	\begin{split}
		&\theta(q)-\hat\theta=\\
		&\mu (q (1-F_H(x(q)))-(1-F_H(\hat x))+(1-\mu)\varphi(1-F_L(\hat x)-(1-q)(1-F_L(x(q))))\\=
		&\mu(F_H(\hat x)-qF_H(x(q))-(1-q))+(1-\mu)\varphi(q+(1-q)F_L(x(q))-F_L(\hat x)).
	\end{split}
\end{equation}
The use of a mechanical gatekeeper is beneficial whenever the above difference is positive. After rearranging, we get the following proposition: 
\begin{proposition}\label{prop:NK is better}
	A mechanical gatekeeper with signal quality $q$ yields a correctness that is higher than the benchmark case if and only if,
	$$
	\frac{\mu(1-F_H(\hat x))+(1-\mu)\varphi(F_L(\hat x)-F_L(x(q)))}{\mu(1-F_H(x(q)))+(1-\mu)\varphi(1-F_L(x(q)))}\le q.
	$$
\end{proposition}
To see  part 1 of Theorem \ref{thm:correctness implications}, we will show that for every $q,$ the left-hand side of Proposition \ref{prop:NK is better} is smaller than one.  That is,
$$
\mu(1-F_H(\hat x))+(1-\mu)\varphi(F_L(\hat x)-F_L(x(q)))\le \mu(1-F_H(x(q)))+(1-\mu)\varphi(1-F_L(x(q))).
$$ 
Rearranging, we get,
$$
F_H(x(q))-F_H(\hat x)\le \frac{1-\mu}{\mu}\varphi(1-F_L(\hat x)).
$$
 By equation \eqref{eq:xq}, for every $q\in[0.5,1],$ we get that $\hat x=x(\frac{1}{2})\ge x(q).$ Therefore, $F_H(x(q))-F_H(\hat x)\le 0,$ thus concluding the proposition's proof.\qed

To prove the remaining parts of the theorem, we examine if
  $$\frac{\mu(1-F_H(\hat x))+(1-\mu)\varphi(F_L(\hat x)-F_L(x(q)))}{\mu(1-F_H(x(q)))+(1-\mu)\varphi(1-F_L(x(q)))}<\frac{1}{2}$$
   when $q\approxeq \frac{1}{2}.$
   
    By Proposition \ref{prop:NK is better}, if this condition holds, than the addition of a gatekeeper of such low quality improves correctness. If the condition does not hold, adding the low-quality gatekeeper hurts correctness.

Let  $q=\frac{1}{2}+\varepsilon.$ The condition in Proposition \ref{prop:NK is better} can now be written as,
 	$$
 \frac{\mu(1-F_H(\hat x))+(1-\mu)\varphi(F_L(\hat x)-F_L(x(q)))}{\mu(1-F_H(x(q)))+(1-\mu)\varphi(1-F_L(x(q)))}< \frac{1}{2}+\varepsilon.
 $$
Recall that $\lim_{\varepsilon\rightarrow 0}F_\omega(x(q))=F_\omega(\hat x);$ hence, in the limit for an arbitrarily small $\varepsilon,$ correctness improves only whenever
$$
 \frac{\mu(1-F_H(\hat x))}{\mu(1-F_H(\hat x))+(1-\mu)\varphi(1-F_L(\hat x))}\le \frac{1}{2}.
$$
After rearranging, we get the above condition.
\end{proof}
\section{Strategic Gatekeeping}\label{sec:strategic}


In this section, we investigate whether the mechanical gatekeeper's performance can be improved. Theorem \ref{thm:correctness implications} indicated that introducing a gatekeeper can sometimes reduce the selection process's efficiency, as reflected in its correctness. This issue is now analyzed from the gatekeeper's perspective. The gatekeeper's goal, akin to correctness, increases with the probability of a suitable candidate passing the test. However, the gatekeeper's incentives diverge from our concept of correctness, as the gatekeeper suffers negative outcomes only when a candidate who is a poor fit for the position passes the test. Additionally, correctness is also diminished by the opportunity cost, which deters a highly suitable candidate from applying later. Next we examine whether introducing variability could also be advantageous for the gatekeeper.

We will assume that the following condition on the gatekeeper's utility to avoid trivial cases holds,
	\begin{equation}\label{eq:nk is eq}
		\begin{split}
	&\mu q (1-F_H(x(\sigma_{MK}(q))))-d\varphi(1-\mu)(1-q)(1-F_L(x(\sigma_{MK}(q))))\\
	&>0>\\
	&\mu (1-q) (1-F_H(x(\sigma_{MK}(q))))-d\varphi(1-\mu)q(1-F_L(x(\sigma_{MK}(q)))).
	\end{split}
	\end{equation}
In equation \eqref{eq:nk is eq}, if the first inequality is reversed, then the gatekeeper's dominant strategy is to accept the applicant regardless of its signal. If the second inequality is reversed, the gatekeeper always rejects the applicant. Thus, we focus on the case in which both inequalities are strict. 

We model strategic gatekeeping by allowing the gatekeeper to use a mixed strategy between utilizing the mechanical strategy with probability $1-\sigma$ and, with probability $\sigma$, allowing the applicant (i.e., the candidate who chooses to apply) to take the test without filtering.\footnote{Note that this is equivalent to giving a gatekeeper with a low signal the ability to play a mixed strategy.}   In the following proposition, we show that this method indeed counters the adverse effects on candidate self-selection behavior from Theorem \ref{thm:selfSelect}.

\begin{proposition}\label{prop:selfselectFix}
	For every $q,\mu,F_L,F_H,S$ that satisfy equation \eqref{eq:nk is eq} and a mixed strategy $\sigma,$ the following inequality is satisfied,
	$$
	\frac{\partial x}{\partial \sigma}>0,
	$$
	where $x(\sigma)$ is the indifferent type of candidate when the probability of no gatekeeping is $\sigma.$
\end{proposition}
\begin{proof}
	Let $\sigma$ be the probability at which the candidate faces no gatekeeping. If she chooses to apply, her expected utility will be,
	$$
	\frac{\mu p(s)(q+\sigma(1-q))(1-\gamma)-(1-\mu)(1-p(s))((1-q)+\sigma q)(\gamma-\varphi\alpha)}{\mu p(s)+(1-\mu)(1-p(s))}.
	$$
	The candidate will be indifferent whenever the following equality holds,
	$$
	x(\sigma)=\frac{(\gamma-\varphi\alpha)(1-\mu)((1-q)+\sigma q)}{(\gamma-\varphi\alpha)(1-\mu)((1-q)+\sigma q)+(1-\gamma)\mu(q+\sigma(1-q))}.
	$$
	The derivative of $x(\sigma)$ is thus,
	\begin{equation}\label{eq:x_diff_sigma}
	\frac{\partial x}{\partial \sigma}=\frac{(\gamma-\varphi\alpha)(1-\mu)(1-\gamma)\mu(2q-1)}{((\gamma-\varphi\alpha)(1-\mu)((1-q)+\sigma q)+(1-\gamma)\mu(q+\sigma(1-q)))^2}.
	\end{equation}
	Recall that the expression above is positive for every $q\in(0.5,1].$
\end{proof}
  Proposition \ref{prop:selfselectFix} shows that strategic gatekeeping can counter candidate self-selection. This, however, comes at a cost as the process dismisses the gatekeeper's signal with positive probability.  
 We now evaluate if the gatekeeper can benefit from using this strategy.

   In Theorem \ref{thm:NK_best}, we show that there are cases in which even the gatekeeper benefits from diverting from the mechanical strategy. In fact, we prove that when the prior is sufficiently high, the gatekeeper always benefits from such strategic behavior. 
   
   \begin{theorem}\label{thm:NK_best}
   	 There exists $\bar\mu$ such that for all $\mu\ge\bar \mu,$ the probability $\sigma$ is strictly positive in every equilibrium. 
   \end{theorem}
\begin{proof}
	The gatekeeper's expected utility from any strategy $\sigma$ can be written as,
	$$U_K(\sigma)=\mu (1-F_H(x(\sigma)))(q+\sigma(1-q))-d\varphi(1-\mu)(1-F_L(x(\sigma)))((1-q)+q\sigma).$$ 
	In the theorem, we claim that $\sigma>0.$ That is, when public belief is sufficiently high, playing the mechanical gatekeeper's strategy is never optimal. Assume, in contrast, that $\mu$ is arbitrarily close to one, yet $\sigma=0$ is an equilibrium. This entails that $\frac{\partial U_K}{\partial \sigma}|_{\sigma\rightarrow 0}\le 0.$
	
	We calculate the derivative of $U_K(\sigma)$ as follows,\footnote{We abuse notation slightly here and refer to $x(q)$ as defined by equation \eqref{eq:xq} rather than calling it $x(\sigma=0).$}
	\begin{equation*}
		\begin{split}
				&\frac{\partial U_K}{\partial \sigma}=\\
				&\mu(1-q)(1-F_H(x(\sigma)))-\mu(q+\sigma(1-q))f_H(x(\sigma))\frac{\partial x}{\partial \sigma}\\
				&-d\varphi(1-\mu)q(1-F_L(x(\sigma)))+d\varphi(1-q+\sigma q)f_L(x(\sigma))\frac{\partial x}{\partial \sigma} 
		\end{split}
	\end{equation*}
       and thus $\frac{\partial U_K}{\partial \sigma}|_{\sigma\rightarrow 0}\le 0$ whenever
	\begin{equation}\label{eq:diff_keeper}
		\begin{split}
	\mu(1-q)(1-F_H(x(q)))-d\varphi(1-\mu)q(1-F_L(x(q)))&\le\\
	(\mu q f_H(x(q))-d\varphi(1-q)(1-\mu)&f_L(x(q)))\frac{\partial x}{\partial \sigma}.
	\end{split}
	\end{equation}
	To see the contradiction, recall that $\mu$ is arbitrarily close to one. Therefore, the left-hand side of \eqref{eq:diff_keeper} is arbitrarily close to $1-q,$ while by equation \eqref{eq:x_diff_sigma}, the expression $\frac{\partial x}{\partial \sigma}$ is arbitrarily close to zero. 
\end{proof}
	At first glance, Theorem \ref{thm:NK_best} may seem to be intuitive for anyone who is well versed in Bayesian analysis. When disregarding candidate self-selection,  whenever the prior is above the gatekeeper's signal quality, it is in the gatekeeper's best interest to disregard its signal and follow the public belief.  In our model, however,  any increase in the initial prior is balanced by a less selective candidate's behavior. When proving Theorem \ref{thm:NK_best},  we account for this behavior as well.  We show that when $\mu$ is sufficiently close to one, $\sigma=0$ is never a local maximum. Thus, cannot be a global maximum. As for the opposite direction, note that even if $\sigma=0$ is a local maximum, there can still be a case where the global maximum is reached for some  $\sigma\in(0,1].$ Therefore the exact characterization of a $\bar{\mu}$ above which mechanical gatekeeping is no longer an equilibrium strategy depends on the exact information structure.

\section{The Biased Gatekeeper Effect}\label{sec:biasedGK}

To examine the effects of biased gatekeepers on candidate self-selection, we will fine-tune////tweak the model presented in Section \ref{sec:model}. Consider a scenario in which two candidates (Alice and Bob, or $A$ and $B,$ respectively.) are competing for the same job. We have now broadened the range of potential natural states to $\Omega=\{\{L,L\},\{H,L\},\{L,H\},\{H,H\}\}.$ We assume that the agents' types are drawn independently, and we denote Alice's prior by $\mu_A=Pr(\omega_1=H)$ and Bob's prior by $\mu_B=Pr(\omega_2=H).$  Alice and Bob have the same information structure, namely, $F_\omega,S_C,$ but due to historic reasons, we assume that the gatekeeper's signal is more accurate when Bob applies, that is, $q_K^B>q_K^A.$ The revised game proceeds as follows. Nature draws a type for Alice and Bob, and each receives a private signal about her/his type and decides whether to apply simultaneously. For each applicant (i.e., candidate who decides to apply), the gatekeeper receives a binary signal and decides whether to reject her (or him). The applicant's cost $\gamma_{i\in\{A,B\}}$ is incurred only if she passes the gatekeeper. For the sake of simplicity, we assume that $\varphi=0,$ that is, if an applicant passes the gatekeeper, her (or his) type is revealed. Furthermore, we assume that the gatekeeper's cost parameter $d$ is sufficiently low, and therefore, the firm will hire at least one candidate (if any applied). 
The firm prefers to hire an applicant who is highly compatible with the position/firm (type $H$). However, if no highly compatible candidates apply for the position, then the gatekeeper will resort to hiring a candidate with lower compatibility (type $L$). In cases where both applicants are of equal ability, the firm chooses between them with equal probability (this assumption is later adjusted to specify that, in these instances, the probability of choosing Alice is $\rho$ within the range $(0,1)$).

\subsection{Analysis and Results for the Biased Gatekeeper}

When solving the game, we will focus on Perfect Bayesian Equilibrium. To focus on the bias effect, we will assume that the gatekeeper is mechanical, i.e., it follows its signal. Let $x_i$ for $i\in\{A,B\}$ denote the lowest subjective quality for which candidate $i$ chooses to apply. We denote the probability that applicant $i\in\{A,B\}$ of type $\omega\in\{H,L\}$ passes the gatekeeper and is hired by $\varphi^\omega_i.$\footnote{This probability will play a role similar to that of the noisy test from Section \ref{sec:model}} An equilibrium of the revised game will therefore be a pair $x_A,x_B$ that solves the following equation pair.
\begin{equation}\label{eq:bias_indif}
\frac{x_i}{1-x_i}=\frac{1-\mu_i}{\mu_i}\frac{1-q_i}{q_i}\frac{\gamma_i-\varphi^L_i}{\varphi^H_i-\gamma_i},
\end{equation}
where
\begin{eqnarray*}
\varphi^H_i&=\mu_j(F_H(x_j)+(1-F_H(x_j))((1-q_j)+\frac{q_j}{2}))+(1-\mu_j)1\\
\varphi^L_i&=\mu_j(F_H(x_j)+(1-F_H(x_j))(1-q_j))+(1-\mu_j)(F_L(x_j)+(1-F_L(x_j))(q_j+\frac{1-q_j}{2})).
\end{eqnarray*}
 Equation \eqref{eq:bias_indif} is similar to our description of the candidate's behavior in Section \ref{sec:model}.\footnote{To see why, simply choose the following parameters: $\gamma=\frac{\gamma}{\varphi^H_i}$ and $\varphi=\frac{\varphi^L_i}{\varphi^H_i}.$} Note, however, that the analysis in this case requires greater subtlety due to equilibrium considerations. From now on we will use Alice and Bob ($A$ and $B$) when discussing the bias and its effect while using $i$ and $j$ for the game theoretic analysis.

\subsubsection{Best Response Analysis}
 We begin our analysis by examining the best response of candidate $i$ to a given $q_i,q_j,\mu_i,\mu_j,$ and a (not necessarily equilibrium) strategy $\sigma_j$ of candidate $j.$
Let $\sigma_j$ denote any strategy for candidate $j.$  
In the following lemma, we characterize the structure of candidate $i$'s best response strategy.
\begin{lemma}\label{lem:brThreshold}
For any $\sigma_j,$ candidate $i's$ best response is a threshold strategy, i.e., there exists a unique $x_i\in[0,1]$ such that candidate $i$'s best response to $\sigma_j$ is to apply if and only if she receives a signal $s_i$ that admits a subjective quality $p(s_i)>x_i.$
\end{lemma}
\begin{proof}
We abuse notation slightly and re-define $F_\omega(\sigma_j)$ as the probability that candidate $j$ will apply given that her type is $\omega$ and that she plays strategy $\sigma_j.$\footnote{Note that if $\sigma_j$ is a threshold strategy, both definitions coincide.} Thus, by stochastic dominance, we know that for every $\sigma_j,$ $F_H(\sigma_j)\geq F_L(\sigma_j).$\footnote{This inequality is strict for every strategy $\sigma_j$ in which there exists a non-zero measure of signals for which candidate $j$ applies when she plays strategy $\sigma_j$.} Our lemma follows from equation \eqref{eq:bias_indif} and the MLRP assumption.
\end{proof}
 
 By Lemma \ref{lem:brThreshold}, we can limit our attention to threshold strategies  and denote the strategy threshold by $x_j.$ In the following lemma, we characterize the changes to a candidate's best response strategy due to a ceteris paribus change in the parameters (or behavior) of her competitor.
\begin{lemma}\label{lem:BR}
Let $\gamma_i\in(\varphi^L_i(q_j,\mu_j,x_j),\varphi^H_i(q_j,\mu_j,x_j))$, and let $x^{BR}_i(q_i,\mu_i,q_j,\mu_j,x_j)$ be the lowest signal for which candidate $i$ chooses to apply for the position when playing her best response strategy. The following conditions hold,
\begin{enumerate}[label={(\alph*)}]
\item if $\tilde{\mu}_j>\mu_j$ then $x^{BR}_i(q_i,\mu_i,q_j,\mu_j,x_j)<x^{BR}_i(q_i,\mu_i,q_j,\tilde{\mu_j},x_j).$
\item if $\tilde{x}_j>x_j$ then $x^{BR}_i(q_i,\mu_i,q_j,\mu_j,x_j)<x^{BR}_i(q_i,\mu_i,q_j,\mu_j,\tilde{x_j}).$
\item if $\tilde{q}_j>q_j$ and $\mu_j>\frac{1-F_L(x_j)}{1-F_L(x_j)+1-F_H(x_j)}$ then $x^{BR}_i(q_i,\mu_i,q_j,\mu_j,x_j)<x^{BR}_i(q_i,\mu_i,\tilde{q}_j,\mu_j,x_j).$
\end{enumerate}
\end{lemma}
\begin{proof}
Note that $\frac{\gamma_i-\varphi_i^L}{\varphi^H_i-\gamma_i}=\frac{\varphi^H_i-\varphi^L_i}{\varphi^H_i-\gamma_i}-1\triangleq\delta(\mu_j,q_j,\gamma_i).$ 
Next we rearrange the above  $\varphi^L_i$ and $\varphi^H_i$ and get,
\begin{align}
&\varphi^H_i=1-\frac{\mu_j(1-F_H(x_j))q_j}{2}\\
&\varphi^H_i-\varphi^L_i=\frac{\mu_j q_j(1-F_H(x_j))+(1-\mu_j)(1-q_j)(1-F_L(x_j))}{2}\triangleq d(\mu_j,q_j,\gamma_i).\label{eq:varphiL}
\end{align}
From the above, one can easily see that, when a candidate competes against a candidate with a better reputation (a higher $\mu_j$), or against a candidate for whom the gatekeeper's signal is more accurate (higher $q_j$), or against one who behaves less selectively (lower $x_j$), the value of $\varphi^H_i$ decreases. The overall effect on $x_i$ is also determined by $\varphi^L_i,$ and thus is unclear. 

To prove $(a)$, note that $$\frac{\partial d}{\partial \mu_j}=\frac{1}{2}(q_j(1-F_H(x_j))-(1-q_j)(1-F_L(x_j)).$$ Because $q_j>\frac{1}{2}$ and $F_H\succeq_{FOSD}F_L,$ we get that $\frac{\partial d}{\partial \mu_j}>0.$ Since $\frac{\partial \varphi^H_i}{\partial \mu_j}<0$, we get that $\frac{\partial \delta}{\partial \mu_j}>0.~\qed$ The proof of $(b)$ is similar and is therefore omitted.

For the proof of $(c)$, note that
\begin{equation*}
\frac{\partial d}{\partial q_j}=\frac{1}{2}\mu_j(1-F_H(x_j))(1-\frac{1-\mu_j}{\mu_j}\frac{1-F_L(x_j)}{1-F_H(x_j)}).
\end{equation*}
 Recall that  $\frac{\partial\varphi_i^H}{\partial q_j}<0,$ and note that the above expression is positive whenever $\mu_j>\frac{1-F_L(x_j)}{1-F_L(x_j)+1-F_h(x_j)},$ and therefore $\frac{\partial \delta}{\partial q_j}>0.$
\end{proof}
\subsubsection{Equilibrium Analysis}
By Lemma \ref{lem:brThreshold}, we know that candidates will follow a threshold strategy in  all equilibria. We assume that our information structure admits a unique equilibrium.\footnote{We believe that equilibrium uniqueness cannot be proved for the general case and that it depends on the specific signal distributions.} Therefore, there can be three types of equilibria: one in which both candidates opt out, a second in which both candidates apply whenever their their respective signals are sufficiently high, and a third in which one candidate opts out and the other applies whenever her signal is sufficiently high. As we discuss next, it turns out that only the latter two types of equilibria occur. 

\textbf{Opting-out Equilibrium:} Note that as $x_j\rightarrow 1$, both $\varphi_i^H$ and $\varphi_i^L$ approach one. Therefore, for any $\gamma_i,$ there exists $\bar{x}_j$ such that whenever $x_j>\bar{x}_j,$ candidate $i$'s best response is to participate regardless of her signal. In equilibrium, therefore, at most one candidate can opt out. An equilibrium in which one candidate opts out can only occur if $\gamma_j>\varphi^H_j(q_i,\mu_i,x_i=0).$

\noindent\textbf{Two Candidate Equilibrium:} We assume that $$\gamma_j\in(\varphi^L_j(q_i,\mu_i,x_i=0),\varphi^H_j(q_i,\mu_i,x_i=0)),$$ and thus, in equilibrium, both candidates will apply if their respective signals are sufficiently high. In the following theorem, we show how the equilibrium strategy of candidate $i$ is affected by changes in the  parameters of candidate $j.$
\begin{theorem}\label{thm:gkBiasEqlm}
Let $x^*_i(q_i,q_j,\mu_i,\mu_j,\gamma_i,\gamma_j)$ be the threshold of candidate $i$'s equilibrium strategy in the game defined by the parameters $q_i,q_j,\mu_i,\mu_j,\gamma_i,\gamma_j.$ The following conditions hold,
\begin{enumerate}[label={(\alph*)}]
\item if $\tilde{\mu}_j>\mu_j$ then $x^{*}_i(q_i,q_j,\mu_i,\mu_j,\gamma_i,\gamma_j)<x^{*}_i(q_i,q_j,\mu_i,\tilde{\mu_j},\gamma_i,\gamma_j).$
\item if $\tilde{q}_j>q_j$ and $\mu_j>\frac{1-F_L(x_j)}{1-F_L(x_j)+1-F_H(x_j)}$ then $x^{*}_i(q_i,q_j,\mu_i,\mu_j,\gamma_i,\gamma_j)<x^{*}_i(q_i,\tilde{q}_j,\mu_i,\mu_j,\gamma_i,\gamma_j).$
\item if $\tilde{\gamma}_j>\gamma_j$ then $x^{*}_i(q_i,q_j,\mu_i,\mu_j,\gamma_i,\gamma_j)<x^{*}_i(q_i,q_j,\mu_i,\mu_j,\gamma_i,\tilde{\gamma_j}).$
\end{enumerate}
\end{theorem}
To better understand the consequences of Theorem \ref{thm:gkBiasEqlm}, we provide here the intuitive details. Consider the case where Bob's prior $\mu_b$ is sufficiently high (i.e., $\mu_B>\frac{1-F_L(x_B)}{1-F_L(x_B)+1-F_H(x_B)}$). Suppose that the gatekeeper's signal accuracy for Bob $q^B_K$ is arbitrarily close to one, while its accuracy for Alice $q^A_K$ is low (e.g., $q^B_K\approx 1$ and $q^A_K\approx 0.5$). By the indifference condition \eqref {eq:bias_indif}, since $q^B_K\approx 1$, Bob will apply unless he receives an extremely low signal ($x_B\approx 0$). Hence, Alice knows that if she applies, she will likely face competition with Bob. In this scenario, there is a high probability that Alice will not be hired even if she highly suitable to the position. Furthermore, with low $q^A_K$, Alice knows that the gatekeeper provides little protection, i.e., her probability of passing screening is the same, regardless of whether her actual fit to the job is high or low. Together, Bob's non-selective behavior and Alice's unreliable gatekeeper signal significantly discourage Alice from participating. This demonstrates how the interaction between bias and strategic incentives can perpetuate underrepresentation.

\begin{proof}
We will show the proof of $(b).$ The proofs of $(a),(c)$ are similar, thus omitted. Assume by contradiction that  $\tilde{q}_j>q_j$ and $\mu_j>\frac{1-F_L(x_j)}{1-F_L(x_j)+1-F_H(x_j)}$ yet $x^{*}_i(q_i,q_j,\mu_i,\mu_j,\gamma_i,\gamma_j)\geq x^{*}_i(q_i,\tilde{q}_j,\mu_i,\mu_j,\gamma_i,\gamma_j) \triangleq x^{*}_i\geq\tilde{x}^{*}_i.$ 

By equation \eqref{eq:bias_indif} and the assumption that $x^{*}_i\geq\tilde{x}^{*}_i$, 
$$
\frac{\tilde{x}_j}{1-\tilde{x_j}}=\frac{1-\mu_j}{\mu_j}\frac{1-\tilde{q}_j}{\tilde{q}_j}\delta(\mu_i,\tilde{q}_i,\gamma_j,\tilde{x}^{*}_i)>\frac{1-\mu_j}{\mu_j}\frac{1-q_j}{q_j}\delta(\mu_i,q_i,\gamma_j,x^{*}_i)=\frac{x_j}{1-x_j},
$$
where $\tilde{x}_j=x^{*}_j(\tilde{q}_j,q_i,\mu_j,\mu_i,\gamma_j,\gamma_i)$ and 
$\delta(\mu_i,q_i,\gamma_j)\triangleq\frac{\gamma_j-\varphi_j^L}{\varphi^H_j-\gamma_j}.$ 

Since $x_j<\tilde{x}_j$ and $q_j<\tilde{q}_j$ and $\mu_j>\frac{1-F_L(x_j)}{1-F_L(x_j)+1-F_H(x_j)},$ by Lemma \ref{lem:BR} $x^{*}_i>\tilde{x}^{*}_i,$ a contradiction.
\end{proof}
We return to our story of Alice and Bob, when $\mu_B$ is sufficiently high, an improvement in the gatekeeper's signal accuracy over Bob induces two effects that cause Alice to behave more selectively.\footnote{If $\mu_j<\frac{1-F_L(x_j)}{1-F_L(x_j)+1-F_H(x_j)},$ the effect of accuracy is determined by the exact information structure $F_H,F_L.$ Although we speculate that our result still hold under some general conditions.} First,  
If Bob faces a the gatekeeper whose signal is more accurate, by Theorem \ref{thm:selfSelect}, he will be less selective. Thus Alice is more likely to find herself in a competitive scenario. Second, her expected utility decreases even in this competitive scenario, as if Bob passes the gatekeeper, he is more likely to be of type $H,$ which decreases Alice's probability of being eventually hired.

In the next section we study the severity of this phenomena using a numeric example.

\section{An Example for the Gatekeeper Effect}\label{sec:example}
As an example, assume that the candidate information structure is as follows,\footnote{See \cite{Ban2020a} for further details.} 
$$
F_H(x)=x^2;F_L(x)=1-(1-x)^2;x\in(0,1).
$$
Assume that $\gamma=0.4, \varphi=0.6, \alpha=0.5,$ and $d=1.$ In figure \ref{fig-prop1} we calculate the condition from Proposition \ref{prop:NK is better} and check whether adding a gatekeeper improves or degrades the decision's correctness.\\

In figure \ref{fig-prop1} we can see the connection between the gatekeeper effect on correctness, the unconditional probability over the states of nature (i.e., the public belief), and the gatekeeper's signal quality. As expected, the gatekeeper positively affects correctness whenever her signal quality is sufficiently high for any value of public belief. One can see this by the green region at the top of Figure \ref{fig-prop1}. The third part of Theorem \ref{thm:correctness implications} tells us that the gatekeeper effect can, at times, be positive, even  when its signal quality is extremely low. In Figure \ref{fig-prop1} we see that this occurs whenever the public belief is sufficiently low.  Finally, we see that as the public belief increases, a positive gatekeeper effect emerges only for increasingly higher signal qualities. This finding confirms our third result, which states that strategic gatekeeping may improve performance for sufficiently high public beliefs.  We provide the inverse image in Figure \ref{fig-prop2}, where we calculate the threshold above which playing the ``mechanical gatekeeper" strategy is no longer an equilibrium. As one can see, it is increasing in the signal quality. Additionally, when the gatekeeper's disutility from hiring type $L$ ($d$ increases), the minimal threshold above which mechanical gatekeeper is no longer an equilibrium increases. A reduction in the candidate's cost, on the other hand, lowers this threshold making strategic gatekeeping beneficial even for lower priors.

As we can see in Figure \ref{fig-prop2}, if $d\le 1,$ the threshold $\bar{\mu}$ is below the identity line. This is to be expected as when $\mu$ approximately equals $q_K,$ the strength of the public belief cancels out a negative gatekeeper signal. This intuition has merit even when considering the strategic behavior of the candidate. 
\begin{figure}
\centering
\begin{minipage}{\textwidth}
\begin{subfigure}[t]{0.5\textwidth}
    \begin{overpic}[width=0.75\textwidth]{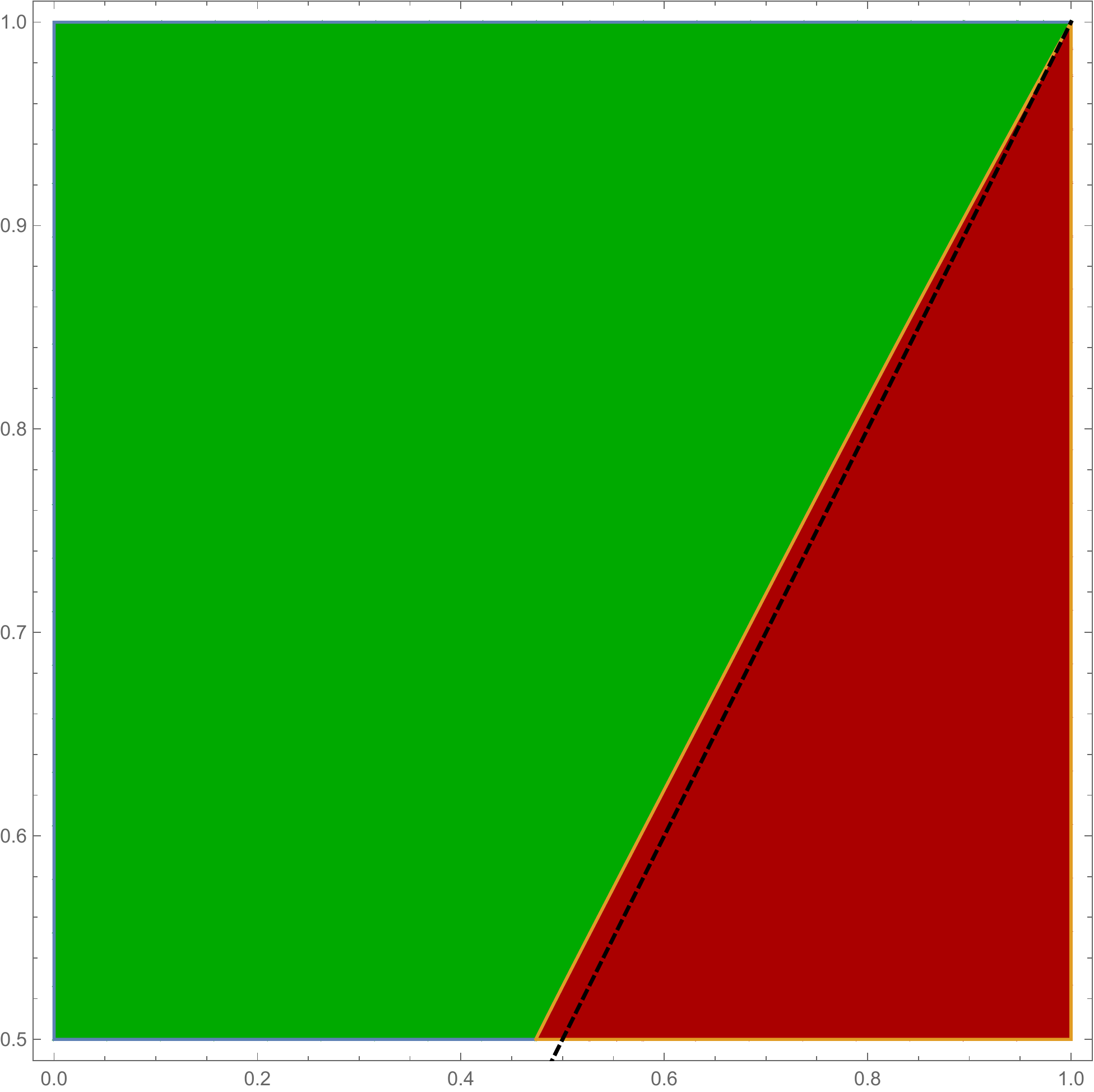}
      \put(-5,50){\makebox(0,0){\rotatebox{90}{\tiny $q_K$}}}
      \put(53,-1){\tiny$\mu$}
    \end{overpic}
\caption{The Mechanical Gatekeeper effect on Correctness.}
\label{fig-prop1}
\end{subfigure}%
\begin{subfigure}[t]{0.5\textwidth}
\scalebox{0.5}{\input{plotSGK.tex}}
\caption{$\bar{\mu}$ for strategic gatekeeping}
\label{fig-prop2}
\end{subfigure}
\end{minipage}
\caption{The Gatekeeper Effect}
\subcaption*{\tiny \textbf{Figure \ref{fig-prop1}}: The public belief $\mu=Pr(\omega=H)$ is on the horizontal axis, the gatekeeper's signal quality is on the vertical axis. The dashed line is $\mu=q_K.$ The rest of the parameters are $\gamma=0.4,\varphi=0.6,\alpha=0.5,$ and $d=1.$  By Proposition \ref{prop:NK is better}, the introduction of a gatekeeper improves correctness in the green regions and harms it in the red regions. \textbf{Figure \ref{fig-prop2}:} The public belief $\mu=Pr(\omega=H)$ is on the vertical axis. The gatekeeper's signal quality is on the horizontal axis. The rest of the parameters are $\gamma=0.4,\varphi=0.6,\alpha=0.5,$ and $d=1.$  The dashed line is the identity line. By Theorem \ref{thm:NK_best}, for every $q,$ there exists a threshold $\bar{\mu}$ such that whenever the public history is above it, the mechanical gatekeeper is no longer an equilibrium. The solid blue line describes this threshold.}
\label{fig:BGKex}
\end{figure}

\FloatBarrier
\subsection{An Example of a Biased Gatekeeper Effect}
In Theorem \ref{thm:gkBiasEqlm} we prove that facing an opponent for which the gatekeeper's signal is more accurate increases the candidate's self-selectivity. In Figure \ref{fig:BGKex} we study the severity of this phenomenon. Figure \ref{fig:BGKeq} describes the equilibrium thresholds of both candidates as the signal accuracy of candidate $j$ changes. From this figure one can see that the candidate for which the accuracy is better, will always be less selective. Furthermore, this gap in self-selectivity increases with the gap in accuracy. 
 One way to mitigate the biased gatekeeper effect is to use an affirmative action. That is to prefer selecting candidates from an under represented group. 
 
 In Section \ref{sec:biasedGK} we assume that if both candidates apply, pass the gatekeeper, and are discovered to be of equal type, than the firm chooses among them with equal probability. To allow for affirmative action we simply release this assumption and assume that in this case Alice is hired with probability $\rho\in[0,1].$ Another possible, stricter, affirmative action policy is that if both candidates apply and pass the gatekeeper, the firm hires Alice with with probability $\rho\in[0,1],$ (i.e., even if her type is lower than Bob's). We call the latter type-invariant affirmative action. In Figure  \ref{fig:BGKProb} we demonstrate the use of two possible variations for such affirmative action. The dotted line describes a policy in which the candidate who faces the less accurate gatekeeper  is preferred, independently of her type. The dashed lines describe a policy which prefers her only if both candidates are revealed to be of  the same type. As we can see in Figure \ref{fig:BGKProb}, such a policy will, naturally, increase the probability of such candidate of being selected (blue lines) and the probability of her being selected if she is of type $H$ (red lines). The policies differ however in its effect over the system's short term efficiency. While type-invariant affirmative action decreases both the probability of hiring high quality candidate and the possibility of hiring the best available candidate, an affirmative action which considers the candidate ex-post type will have little effect over the overall probability of hiring any $H$ type candidate. The interesting part is that such a policy will increase the probability of hiring the best available candidate (the orange line). 

\renewcommand{\thefigure}{2}
\begin{figure}
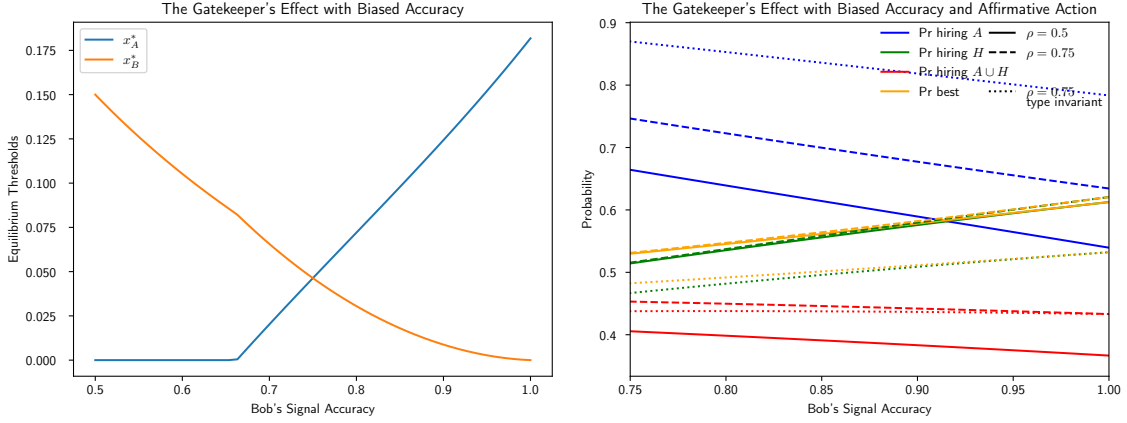

\centering
\begin{minipage}{\textwidth}
\begin{subfigure}{0.5\textwidth}
\scalebox{0.5}{\input{plot.tex}}
\caption{Candidates' Thresholds in Equilibrium.}
\label{fig:BGKeq}
\end{subfigure}%
\begin{subfigure}{0.5\textwidth}
\scalebox{0.5}{\input{plotAA-2.tex}}
\caption{Hiring probability with $\rho\in\{.5,.75\}.$}
\label{fig:BGKProb}
\end{subfigure}
\end{minipage}
\caption{The Biased Gatekeeper Effect for $\mu_A = \mu_B=0.5, q^A_K = 0.75, \gamma_A = \gamma_B = 0.6$ and varying $q^B_K.$}
\subcaption*{\tiny Figure \ref{fig:BGKeq} describes the variation in the thresholds of agents equilibrium strategy.  Figure \ref{fig:BGKProb} describes the distribution over the resulting allocation.  The blue line describes the probability of hiring AliceCandidate $i$ (with the lower signal). The red line is the probability of hiring AliceCandidate $i$ only when her type is $H.$ The green line is the probability of hiring any candidate of type $H.$ The orange line is the probability of hiring the best available candidate. The parameters are $\mu_A = \mu_B=0.5, q_K^A = 0.75, \gamma_A = \gamma_B = 0.6$ and varying $q_K^A.$ Solid lines describe equilibrium where $\rho=0.5.$ The dashed lines are with affirmative action that is applied only when both agents share a type, i.e.,  $\rho=0.75.$, dotted lines describe a type-invariant affirmative action, i.e.,  Alice is hired at $\rho=0.75,.$ if she passes the gatekeeper.}
\label{fig:BGKex}
\end{figure}
\FloatBarrier
\section{Applications Beyond Hiring}\label{sec:discu}
Next we discuss the implications of our results in the context of  the role of reputation in the academic peer-review process.\footnote{I would like to thank an anonymous reviewer for raising this issues.}

Another example of the gatekeeper effect can be found in the process of academic publications. This process is costly to all involved parties.  When deciding whether to publish an article, the journal's editor invests resources in locating suitable reviewers, which perform the task of assessing the proposed publication's quality. The author incurs the submission costs and the alternative waiting costs as she can not simultaneously submit her work to a different venue.  The fit of the proposed article to a specific venue depends on immeasurable characteristics such as writing quality, scientific contribution, etc. Therefore, both author and editor can only estimate the probability in which the article will pass the reviewing process and the probability it will become significant post-publication. To reduce selection costs, the editor usually performs a brief review of the paper and chooses whether to initiate a desk rejection or continue with the rigorous peer-review process.  The availability of a desk rejection has two contradicting effects on the efficiency of the selection process. On the one hand, desk rejection protects both the author and the editor against unnecessary costs. On the other hand, it may induce a less careful behavior by the authors, as now, their costs will materialize only in the event that the editor finds the paper worthy, hence, have better odds of passing the review process.

The peer-review process differs significantly among academic disciplines. For example, in Economics, the process is single-blinded (i.e., reviewers know the identity of the authors but not vice versa). However, in Computer Sciences, the process is mostly double blinded (i.e., reviewers do not know the identity of the authors and vice versa). The role of the gatekeeper in this scenario is played by the editor (or the area chair in conferences, whenever a desk rejection is possible). 
In our model, the difference in methodologies is captured by the common prior. If the process is single-blinded, outstanding researchers will have high priors while novice researchers correspond to a balanced prior. If the process is double-blinded, there is no public information. Thus, the prior is balanced. Our results support the single-blind approach, as the gain from employing a gatekeeper varies with the reputation. Mainly,  correctness decreases when the common prior is sufficiently high.  Intuitively, in practice, one would expect an even more severe effect as there is a strong correlation between researchers with an outstanding reputation and those whose private signals are of high quality.

Additionally, our results for the biased gatekeeper show give case for leniency when deciding over a desk rejection of a paper who is further than the main stream literature usually sent to a specific venue (or by a novice researcher). Such lenient policy, provided that the examined paper is of high quality of course, is beneficial in two ways. First, it will elicit the participation of, on average, higher quality paper from the currently under represented groups, and second, it will benefit the editorial process as the accuracy in assessing such papers will increase.

\section{Conclusion}\label{sec:conc}
A gatekeeper is a common feature in many costly selection processes. Its goal is to reduce overall selection cost by sifting the wheat from the chaff before administering a costly exam. While not without merit, this intuition disregards the gatekeeper's indirect effects on strategic behavior by those who are being selected, i.e., the candidates. This work studies the implications of the gatekeeper's introduction on the resulting decision's quality, while also considering the candidate's behavior.

We introduce a game of incomplete information. A candidate must decide whether to participate in a costly selection process. A gatekeeper must decide whether she allows the candidate to take a costly exam. We find that the presence of a gatekeeper induces less self-selection on behalf of the candidate. Furthermore, the higher the  quality of the gatekeeper signal, the less selective the candidate behavior becomes.  We analyze the consequences of this indirect effect and find that (1) Provided that her signal quality is significantly high, the addition of a gatekeeper improves the probability of the process resulting in a correct decision. (2) Surprisingly, there are cases in which even the addition of a very low-quality gatekeeper improves correctness. This phenomenon occurs whenever the test is sufficiently refining.

We extend our model to account for bias in the gatekeeping process and find evidence that under represented groups will tend to be more selective than those for which the gatekeeper's signal is more accurate. We suggest a remedy for this inefficiency by using affirmative action.

\subsection{Technical Choices and Limitations}
 In the technical side, we make two assumptions about the information structure. First, we describe the gatekeeper's private information as a binary signal structure. Second, we assume that the candidate's signals are unbounded. We argue that these assumptions can be relaxed. 
As for the binary gatekeeper's signal, in a previous version (available upon request), we have assumed a richer gatekeeper's signal structure. Note that Theorem \ref{thm:selfSelect} depends on the ratio $\frac{Pr(H|a_K=1)}{Pr(L|a_K=1)}.$ When the gatekeeper's signals are continuous, our analysis carries through. However, in this case, it will rely on comparing the gatekeeper's indifference signal generated by two separate information structures. Theorem \ref{thm:correctness implications} also carries through, although the exact formulation becomes cumbersome. As for Theorem \ref{thm:NK_best}, when the gatekeeper's signal is continuous, its strategy will follow a threshold structure. Thus, indifference will almost surely not occur.  Finally, a binary classification seems to align with our motivation to examine the utilization of an AI classifier in selection processes.

When we have considered allowing the candidate signal structure to be bounded, two issues emerged. First, our results will carry through in full if the public belief is approximately balanced (i.e., $\mu\approx\frac{1}{2}$). Therefore, releasing the unbounded signals assumption will adversely affect the readability and approachability of our results. A second issue that emerges when signals are bounded is that now multiple equilibria may emerge.

Our goal was to present a model that is tractable on the one hand, yet is sufficiently general and  make as few assumptions as possible about the candidate signal structure.
Note that by assuming specific distributions, as we do in Section \ref{sec:example}, one can calculate the exact premium generated from the gatekeeper's introduction. One can also use a more structured version of our model to analyze self-selection in various settings. Our third result justifies adding noise into selection processes. This choice of strategy must be communicated to the candidates to induce a more selective candidate behavior. By employing such a strategy, a system designer can also gain critical insight into the selection process.  For example, suppose an AI software of known quality plays the role of the gatekeeper. In this case, one can estimate the quality of the test by utilizing a strategic gatekeeper strategy.  The designer can thus compare the performance of applicants that the gatekeeper flagged as highly fitting to those who were low fitting and have taken the test only due to the strategic play. When estimating the gatekeeper's quality, the public signal, and the test quality, the designer can approximate the effect of the candidate's self-selection.

Additionally, we argue that when the distributions are ''well behaved" (that is, smooth and  Lipschitz continuous), the optimal strategy will be one of the pure strategies, i.e., mechanical gatekeeper or no gatekeeper at all. This line of inquiry will require assuming additional structure. Therefore, We leave it to future work. 
\subsection*{Managerial Implications}

This paper highlights several key considerations for managers designing selection processes like hiring. A central finding is that gatekeepers —while filtering candidates— can also discourage applicant self-selection in unseen ways. Our model shows that higher gatekeeper accuracy leads to less selectivity from candidates (Theorem 1), occasionally harming overall process correctness (Theorem 2). This has implications for implementation.

First, managers must account for the indirect effects of gatekeepers on applicant self-selection. Groups facing lower gatekeeper accuracy (e.g. due to bias) apply more selectively (Theorem 4). Firms should proactively address statistical biases in screening without compromising bars. This will not only encourage diversity, but can also aid the success of the recruitment effort as a whole.  Confirming our model, research shows algorithms can mitigate human biases and boost diversity in hiring (Cowgill, 2020).

Counterintuitively, “strategic” gatekeepers can sometimes improve applicant selectivity and process efficacy by occasionally ignoring signals (Theorem 3, Proposition 2). Hence, managers should consider occasionally waiving questionable candidates through gatekeepers to explore applicant quality. Studies similarly find exploration algorithms enhance candidate quality and diversity (Li et al., 2020).

Finally, underrepresented candidates eventual performance offers insight into whether self-selection affects current processes. As we show, subgroups facing lower gatekeeper accuracy should demonstrate higher average performance if hired. Confirming this would signify that enhanced efforts to reduce barriers are warranted.

In summary, gatekeepers have nuanced impacts on applicant behavior and process outcomes. Managers must address perceptions of bias, allow for occasional exploration, and scrutinize subgroup performance. These steps can mitigate adverse self-selection and enhance the quality and diversity of selection processes.


\bibliographystyle{plainnat}
\bibliography{crowdfunding}
\end{document}